\documentclass[conference]{IEEETran}
\usepackage{amsthm}
\usepackage{amssymb,amsmath,amsthm,amsfonts}
\usepackage{paralist}
\usepackage{enumitem}
\usepackage{flushend}
\usepackage{multirow}
\usepackage{titlesec}
\usepackage[T1]{fontenc}
\usepackage{mathptmx}
\usepackage[most]{tcolorbox}

\usepackage[most]{tcolorbox}
\newtcolorbox{mybox}[2][]{%
  enhanced, breakable, colback=white, colframe=black,
  title={#2}, fonttitle=\sffamily\bfseries, boxrule=0.6pt,
  #1
}

\newtheorem{lemma}{Lemma}
\newtheorem{theorem}{Theorem}

\usepackage{tikz}

\usepackage{xcolor,colortbl}
\usepackage{booktabs}
\usepackage{algorithm}
\usepackage[noend]{algpseudocode}
\usepackage{bm}
\usepackage{graphicx}
\usepackage{tikz}
\usepackage{pgfplots}
\usepackage{url}
\usepackage{hyperref}
\usepackage{soul}
\usepackage{tipa}
\usepackage{epsfig}
\usepackage{stackrel}
\usepackage{subfigure}

\renewcommand{\footnotesize}{\fontsize{7.5pt}{8pt}\selectfont}

\pgfplotsset{compat=1.18}

\title{One-Shot Secure Aggregation: A Hybrid Cryptographic Protocol for Private Federated Learning in IoT }
\author{\IEEEauthorblockN{
Imraul Emmaka, 
Tran Viet Xuan Phuong
}
\IEEEauthorblockA{
Department of Computer Science, University of Arkansas at Little Rock, USA\\
Email: ikemmaka@ualr.edu, ptran@ualr.edu}
}
\date{}

\IEEEoverridecommandlockouts
\IEEEpubid{\makebox[\columnwidth]{\footnotesize The source code for this work is available at: \url{https://github.com/emmaka9/secure_fl} \hfill} \hspace{\columnsep}\makebox[\columnwidth]{ }}

\begin{document}
\maketitle
\begin{abstract}
Federated Learning (FL) offers a promising approach to collaboratively train machine learning models without centralizing raw data, yet its scalability is often throttled by excessive communication overhead. This challenge is magnified in Internet of Things (IoT) environments, where devices face stringent bandwidth, latency, and energy constraints. Conventional secure aggregation protocols, while essential for protecting model updates, frequently require multiple interaction rounds, large payload sizes, and per-client costs rendering them impractical for many edge deployments.

In this work, we present Hyb-Agg, a lightweight and communication-efficient secure aggregation protocol that integrates Multi-Key CKKS (MK-CKKS) homomorphic encryption with Elliptic Curve Diffie-Hellman (ECDH)-based additive masking. Hyb-Agg reduces the secure aggregation process to a single, non-interactive client-to-server transmission per round, ensuring that per-client communication remains constant regardless of the number of participants. This design eliminates partial decryption exchanges, preserves strong privacy under the RLWE, CDH, and random oracle assumptions, and maintains robustness against collusion by the server and up to $N-2$ clients.

We implement and evaluate Hyb-Agg on both high-performance and resource-constrained devices, including a Raspberry Pi 4, demonstrating that it delivers sub-second execution times while achieving a constant communication expansion factor of approximately 12× over plaintext size. By directly addressing the communication bottleneck, Hyb-Agg enables scalable, privacy-preserving federated learning that is practical for real-world IoT deployments.

\end{abstract}

\section{Introduction}
The proliferation of Internet of Things (IoT) devices has fueled an unprecedented growth in distributed data generation, enabling a wide range of applications from smart healthcare and intelligent transportation to industrial automation. 
Federated Learning (FL) supportively complements IoT \cite{pywzzzz25, jaikumar2025femt} by allowing devices to collaborate for global model training while keeping the original data, thereby reducing direct privacy risks. FL still faces a significant challenge of communication inefficiency, particularly in resource-constrained environments such as IoT networks \cite{li2025vmfl}. In every training round, the IoT devices should exchange the high-dimensional model updates with a central server, which easily downgrades the latency and energy consumption of those devices.

In fact, a major bottleneck in federated learning lies in its high communication overhead. In each training round, participating clients must transmit full model updates often consisting of millions of parameters to the central server, which then distributes the aggregated model back to all clients. This bidirectional exchange can be prohibitively expensive, particularly when the number of clients or the model size is large. The problem is further exacerbated in IoT environments, where devices operate under strict bandwidth, latency, and energy constraints. Many existing secure aggregation protocols \cite{bonawitz2017ccs,so2022lightsecagg, fereidooni2021safelearn, mansouri2023sok} aggravate this issue by requiring multiple interaction rounds or additional message exchanges for key setup, partial decryption, or dropout recovery, significantly inflating both uplink and downlink traffic. As a result, the cumulative cost of communication can dominate the total training time, limiting the scalability and practicality of FL in real-world deployments. Recently, the work \cite{ma2022privacy} applies the Multi-Key CKKS for data privacy in FL, enabling it in IoT environments. However, \cite{ma2022privacy} requires partial decryption, introducing an extra interaction round: the server requests per-client decryption shares and combines them to recover the final value. Eliminating this stage is desirable to reduce synchronization and bandwidth, enabling a one-shot (single upload per round) protocol that preserves privacy while minimizing communication cost.

\begin{table*}[!ht]
\centering
\caption{The comparison of secure aggregation of Federated Learning using the cryptographic primitive CKKS/MK-CKKS.}
\label{tab:comparison_ckks_transposed}
\setlength{\tabcolsep}{2pt}
\renewcommand{\arraystretch}{1.2}
\footnotesize
\scalebox{0.9}{
\begin{tabular}{|l|c|c|c|c|}
\hline
\textbf{Metric} & \textbf{Hyb-Agg (Ours)} & \textbf{xMK-CKKS\cite{ma2022privacy}} & \textbf{tMK-CKKS\cite{du2023tmkckks}} & \textbf{FedSHE\cite{pan2024fedshe}} \\
\hline
Cryptographic scheme
& Hybrid MK-CKKS + ECDH masking
& Multi-Key CKKS
& Threshold MK-CKKS
& Single-Key CKKS \\
\hline
Rounds per aggregation
& 1 (non-interactive)
& 2 (upload + decrypt shares)
& 2 (upload + threshold decrypt)
& 1 (centralized key) \\
\hline
Comm. overhead (per client/round)
& Constant; \textasciitilde{}12$\times$ expansion
& High; grows with $N$
& Moderate; \textasciitilde{}$O(t)$ per client
& Low; \textasciitilde{}6.6\% of Paillier-FL \\
\hline
Collusion resistance
& Server + $\leq N-2$ clients
& Server + $\leq N-2$ clients
& Server + $\leq t-1$ clients
& Requires trusted key manager \\
\hline
Deployment / hardware
& Raspberry Pi~4 (sub-second)
& Jetson Nano (IoT cluster)
& Simulated (FedML, up to 5k clients)
& Simulated (PC) \\
\hline
\end{tabular}}
\\[0.25em]
\parbox{0.96\linewidth}{\footnotesize\emph{Note:} The comparisons focus on one-shot (non-interactive) aggregation and contrast the cryptographic scheme, rounds per aggregation, communication cost, collusion model, and deployment.}
\end{table*}
In this work, we present Hyb-Agg, a new variant of a secure aggregation protocol that integrates Multi-Key CKKS (MK-CKKS) \cite{cheon2017ckks} homomorphic encryption with Elliptic Curve Computational Diffie-Hellman (ECDH)-based \cite{haakegaard2015elliptic} additive masking to achieve strong privacy guarantees in just three communication rounds. Technically, MK-CKKS enables homomorphic addition over independently encrypted client updates, which also allows the server to compute the global aggregate without accessing individual ciphertexts. This is suitable for FL's aggregation centric workload, as it supports approximate arithmetic on high dimensional model parameters with minimal overhead. Particularly, we employ ECDH-based additive masking, ensuring the partial decryption shares are unlinkable to any client's update, even if the server colludes with multiple clients. Therefore, our design eliminates the need for partial decryption exchanges, enabling all protected data to be transmitted in a single shot per round. This reduces communication overhead while maintaining confidentiality, integrity, and collusion resistance, even when the server collaborates with up to N - 2 malicious clients.

As a result, our cryptographic design of Hyb-Agg is the combination technique of MK-CKKS and ECDH, which not only strengthens security against honest-but-curious and colluding adversaries, but also eliminates multiple decryption, enabling the aggregation process to be completed in one shot, making Hyb-Agg practical for bandwidth-resource-constrained IoT deployments. To demonstrate our advantages among the secure aggregation in FL, we provide detailed comparisons in Table. ~\ref{tab:comparison_ckks_transposed}. For fairness, we restrict attention to schemes built on the CKKS family (single-key CKKS, MK-CKKS, and threshold MK-CKKS), and align metrics across rounds per aggregation, per-client communication, trust/collusion assumptions, and deployment platform.

Our Hyb-Agg combines MK-CKKS with ECDH masking to achieve one-shot (non-interactive) aggregation. As depicted in Table.~\ref{tab:comparison_ckks_transposed}, our per-client payload is constant in $N$ with $\sim 12 \times$ expansion, and it remains secure even if the server colludes with up to $N-2$ clients, while also achieve the same result on the Raspberry Pi 4. xMK-CKKS \cite{ma2022privacy} retains MK-CKKS but requires two rounds, and its per client cost grows with $N$; it offers the same $N-2$ collusion bound, and was evaluated on Jetson Nano. tMK-CKKS \cite{du2023tmkckks} replaces per-client shares with threshold decryption with two rounds, yielding moderate overhead that scales with the threshold $t$ and collusion resistance up to $t-1$; results are mainly from large-scale simulation. FedSHE \cite{pan2024fedshe} uses single key CKKS, enabling one round aggregation with very low expansion ($\sim 6.6 \% $ of Pailier \cite{paillier1999public}-FL), but paying for the cost of trusted key manager and centralized trust. As can be seen in our comparisons, ours achieve the same features of non-interactive one shot aggregation, constant per-client cost, and decentralized collusion resistance of IoT performance, whereas alternatives add an extra round, and introduce $N / t$ dependent overhead to attain lower byte counts.

We summarize our contributions as follows:
\begin{itemize}
    \item We present a lightweight, non-interactive secure aggregation Hyb-Agg scheme combining MK-CKKS homomorphic encryption with ECDH-based masking, optimized for IoT-FL deployments.
    \item We prove our proposed scheme to achieve confidentiality, integrity, and collusion resistance under standard cryptographic assumptions.
    \item Our empirical benchmarks on both powerful and IoT devices demonstrate the scalability and feasibility of the proposed protocol. We also analyze trade-offs in computation, communication, and storage, highlighting Hyb-Agg’s suitability for bandwidth-constrained, heterogeneous IoT networks.
\end{itemize}
\section{Related Work}
 We focus on the privacy challenges of Federated Learning from the existing works, which mitigate the data leakage while preventing multiple attacks. Additionally, we also investigate all the works of data privacy for FL in the IoT environment.
Firstly, we present the privacy risks in the current FL. In fact, the conventional FL circumvents direct data sharing, the transmission of model updates inherently poses significant privacy risks. Recent research \cite{zhu2019deep,geiping2020inverting} has demonstrated that these parameter updates can inadvertently leak sensitive information regarding the local datasets. Specific privacy threats include: 
(i) Model Inversion Attacks: adversaries can exploit the correlation between the model gradients (or parameter updates) and the underlying data samples. For example, Zhu et al. \cite{zhu2019deep} showed successful reconstruction of private training images from shared gradients, posing severe privacy concerns for clients in FL scenarios. 
(ii) Membership Inference Attacks:
Shokri et al. \cite{shokri2017membership} and Nasr et al. \cite{nasr2019comprehensive} demonstrated that adversaries could infer whether a particular data sample was part of a client's training dataset, directly violating data confidentiality. 
(iii) Gradient Leakage: the fine-grained nature of gradient updates allows attackers to reconstruct raw or partial data points via optimization-based inversion methods. Geiping et al. \cite{geiping2020inverting} validated that even small parameter updates could suffice for precise reconstruction of training examples, emphasizing vulnerabilities inherent in baseline FL. These attacks demonstrate that simply avoiding the exchange of raw data is insufficient to guarantee robust privacy protection, highlighting the urgency of incorporating more advanced cryptographic protocols into the federated learning paradigm.
\subsection{Privacy Threats in Federated Learning}
Federated learning (FL) removes the need to centralize raw data, but exchanging model parameters still risks leakage of private information. A line of work shows that gradients or parameter deltas can be inverted to reconstruct sensitive training samples (``deep leakage from gradients'') or otherwise expose features of the local datasets~\cite{zhu2019deep,geiping2020inverting}. In parallel, membership inference attacks can reveal whether a particular record contributed to training~\cite{shokri2017membership,nasr2019comprehensive}. These results establish that avoiding raw data transfer is insufficient; FL requires privacy-preserving protocols that limit the server’s view to \emph{only} aggregates of many clients’ updates.

\subsection{Cryptographic Secure Aggregation Protocols}
\textbf{Classical masking.}
Bonawitz \emph{et al.} introduced \emph{SecAgg}, the first practical secure aggregation for cross-device FL~\cite{bonawitz2017ccs}. Clients establish pairwise secrets (e.g., via Diffie--Hellman) and mask their updates so that masks cancel when summed; dropout resilience is achieved using Shamir secret sharing. SecAgg offers strong privacy in the honest-but-curious model and robustness to dropouts, but requires multiple message rounds (about five per FL round) and per-client communication that scales linearly in the number of participants, which can be challenging for large IoT cohorts.

\textbf{Faster secret sharing.}
\emph{FastSecAgg} reduces computation and round complexity using an FFT-based multi-secret sharing construction while preserving information-theoretic privacy and tolerance to a bounded fraction of (possibly adaptive) dropouts~\cite{kadhe2020fastsecagg}. Communication remains on the same order as SecAgg but with substantially lower latency and CPU overhead in practice.

\textbf{One-shot aggregate-mask recovery.}
\emph{LightSecAgg} encodes client masks so the server can reconstruct the \emph{sum of masks} in one step, avoiding per-dropout seed recovery~\cite{so2022lightsecagg}. This yields markedly lower overhead under high dropout, reduces active rounds (down to $\sim$3 in synchronous settings), and---distinctively---extends to asynchronous FL, where clients upload at different times.

\textbf{Seed-homomorphic PRGs.}
\emph{SASH} replaces dense pairwise seeding with a seed-homomorphic PRG so that mask combination is algebraic~\cite{liu2022sash}. The key benefit is that per-client overhead is linear in model size but \emph{independent of the number of clients}, and dropout handling requires no extra communication. Empirical results show large efficiency gains over SecAgg at scale.

\textbf{Near-linear scalability via grouping/coding.}
\emph{Turbo-Aggregate} breaks the quadratic communication barrier by grouping clients and employing erasure coding and a multi-group circular masking schedule; total communication scales as $O(N\log N)$ rather than $O(N^2)$, with tolerance to high (non-adaptive) dropout rates~\cite{so2021turboaggregate_jsait}.

\textbf{Three-round single-server designs.}
\emph{MicroSecAgg} streamlines single-server secure aggregation to three rounds total and reduces per-user costs to polylogarithmic in $N$, while server computation is linear in $N$~\cite{guo2024microsecagg_pets}. This substantially lowers round-trip latency and improves robustness in dynamic client participation scenarios common in IoT.

\subsection{Communication Efficiency in Secure Aggregation for IoT}
In cross-device IoT settings, bandwidth and energy constraints make round complexity and per-client payload size decisive. Classical SecAgg’s multi-round flow and $O(N)$ per-client mask traffic can dominate training time when $N$ is large. Recent protocols target these bottlenecks from complementary angles:
(i) \emph{Reduce rounds:} SAFELearn\cite{fereidooni2021safelearn} designs the FL based on Multi-Party Computation and Full Homomorphic Encryption by offering two communication rounds. Then, LightSecAgg and MicroSecAgg compress interactive steps to $\sim$3 rounds (or effectively one-shot uploads for clients), shrinking radio-on time and failure windows~\cite{so2022lightsecagg,guo2024microsecagg_pets}.
(ii) \emph{Remove $N$-dependence:} SASH’s seed-homomorphic masking yields per-client communication that is effectively $O(d)$ (model dimension) with no growth in $N$~\cite{liu2022sash}.
(iii) \emph{Near-linear total traffic:} Turbo-Aggregate achieves $O(N\log N)$ system-level communication via grouping/coding~\cite{so2021turboaggregate_jsait}.
(iv) \emph{Bounded-dropout efficiency:} FastSecAgg and LightSecAgg avoid per-user recovery chatter by design, so costs hinge on the number of \emph{active} clients rather than the number of dropouts~\cite{kadhe2020fastsecagg,so2022lightsecagg}.

Overall, the current state of the art demonstrates that secure aggregation can be made practical for bandwidth- and energy-constrained IoT deployments by minimizing rounds, eliminating $N$-dependent per-client traffic, and engineering dropout-resilient mask handling.

\section{Preliminaries}
\label{sec:preliminaries}

\subsection{Federated Learning}
\label{subsec:fl}
Federated Learning (FL) is a decentralized machine learning paradigm that enables multiple clients to collaboratively train a shared global model without transmitting their raw local data to a central server~\cite{mcmahan2017communication}. Formally, consider a setting with a central server S and a group of N clients, $\mathcal{U}=\{U_{1}, U_{2}, ..., U_{N}\}$. Each client $U_{i}$ holds a private local dataset $D_{i}$. The objective is to minimize a global loss function $L(\theta)$, which is the weighted average of the local loss functions $L_i(\theta)$:
$$ \min_{\theta} L(\theta) = \sum_{i=1}^{N} \frac{|D_i|}{|D|} L_i(\theta) $$
where $D = \bigcup D_i$. This is typically achieved using the \textbf{Federated Averaging (FedAvg)} algorithm, which proceeds in iterative rounds. In each round $t$, the server distributes the current global model $\theta_t$ to a subset of clients. Each client then trains this model on its local data to compute an updated model $\theta_{t+1}^{(i)}$, which is sent back to the server for aggregation into the new global model, $\theta_{t+1}$.

\subsection{Threat Model and Security Goals}
\label{subsec:threat_model}
We define the adversarial capabilities, trust assumptions, and security goals that our secure aggregation protocol must satisfy.

\subsubsection{Adversarial Capabilities}
We consider three classes of adversaries:
\begin{itemize}
    \item \textbf{Honest-but-Curious Server:} The central server S faithfully follows the protocol specification but may record and analyze all messages it receives in an attempt to infer any client's private model update.
    \item \textbf{Honest-but-Curious Clients:} A subset of clients may collude with the server, sharing their own secret keys, masks, or other protocol messages to recover another client's update.
    \item \textbf{External Eavesdropper:} A network-level adversary who can intercept all messages exchanged between clients and the server. We assume all channels are authenticated and integrity-protected (e.g., via TLS).
\end{itemize}

\subsubsection{Trust Assumptions}
\begin{itemize}
    \item \textbf{Key Setup:} We assume a common reference string (CRS) is generated honestly once and shared with all parties.
    \item \textbf{Honest Majority:} At most $k < N-1$ clients may collude with the server. We require at least two non-colluding clients to remain honest to ensure masks cancel correctly.
\end{itemize}

\subsubsection{Security Goals}
Our protocol is designed to achieve the following properties:
\begin{itemize}
    \item \textbf{Confidentiality:} No party (server or colluding clients) learns any individual client's model update $m_i$ in the clear. The server only learns the global sum $\sum m_i$.
    \item \textbf{Integrity:} The aggregated result computed by the server is guaranteed to be the exact arithmetic sum of all honest clients' plaintexts.
    \item \textbf{Collusion Resistance:} Even if the server colludes with up to $k \le N-2$ clients, the updates of the remaining honest clients remain confidential.
\end{itemize}

\subsection{Cryptographic Building Blocks}
\label{subsec:crypto_primitives}

\subsubsection{Ring Learning with Errors (RLWE) Assumption}
The security of our homomorphic encryption scheme is based on the RLWE assumption. Let $n$ be a power of two and $q$ be a prime modulus, defining a polynomial ring $R_q = \mathbb{Z}_q[X]/(X^n+1)$. Let $\chi$ be a narrow discrete Gaussian error distribution over this ring. The RLWE assumption states that it is computationally infeasible to distinguish between pairs of the form $(a, b = s \cdot a + e)$—where $a \leftarrow U(R_q)$ is uniform and $s, e \leftarrow \chi$—and pairs $(a, u)$, where both $a$ and $u$ are uniformly random in $R_q$.

\subsubsection{Multi Key Homomorphic Encryption}
MK-CKKS~\cite{Chen2019MKCKKS} is a multi-key variant of the approximate-arithmetic CKKS scheme~\cite{cheon2017ckks}, supporting homomorphic addition across independently-keyed ciphertexts. In order to only secure the data aggregation,  we leverage the additive homomorphism of MK-CKKS in our federated-learning pipeline. In our proposed work, each client~$U_i$ holds a secret key $s_i$ and public key $b_i$, encrypts under $(s_i,b_i)$, and the server homomorphically sums and partially decrypts to recover only the global sum—no client’s update ever appears in the clear.

Let $n$ be a power of two and $q$ a large prime.  Define the cyclotomic ring:
\begin{align*}
    \mathcal{R} = \mathbb{Z}[X]/(X^n + 1),\quad
    \mathcal{R}_q = \mathbb{Z}_q[X]/(X^n + 1),    
\end{align*}
where all additions and multiplications are implicitly \emph{mod $q$}.  Fix two discrete-Gaussian error distributions $\chi,\psi$ over~$\mathcal R$, and rely on the RLWE assumption that samples:
$$
  (a,b = s\cdot a + e)\;\in\mathcal R_q^2,
  \quad
  a\!\leftarrow U(\mathcal R_q),
  \quad s\!\leftarrow\chi,
  \quad e\!\leftarrow\psi
$$
are indistinguishable from uniform. Then, we represent the construction of MK-CKKS by the following descriptions:
\begin{itemize}
    \item \textbf{Setup \& Key Generation:} all parties agree on $(n,q,\chi,\psi)$ and sample once $ a \xleftarrow{\$} U(\mathcal R_q)$. Each client $U_i$ generates $  s_i \xleftarrow{\$}\;\chi,\quad e_i \xleftarrow{\$}\;\psi,$ and sets:
    \begin{align*}
      \text{Secret key:}\quad & s_i\in R_q,\\
      \text{Public key:}\quad 
      & b_i = -\,s_i\!\cdot a \;+\; e_i\,\in R_q,
      \quad 
      \text{pk}_i=(b_i,a).
    \end{align*}

    \item \textbf{Encryption:} Client $U_i$ with $\text{pk}_i=(b_i,a)$ encodes $m\leftarrow\mathsf{Encode}(x)$, then samples:
\[
  v \xleftarrow{\$}\chi,
  \quad
  e_0,e_1 \xleftarrow{\$}\psi,
\]
and outputs ciphertext $ct_i = (c_0^i, c_1^i)$.

\item \textbf{Decryption: } 
An honest client $U_i$ holding $s_i$ computes dot product between $sk_i = (1, s_i)$ and $ct_i = (c_0^i, c_1^i)$
\[
  c_0^i + c_1^i\cdot s_i
  = m_i + \underbrace{v\cdot e_i + e_0 + e_1\cdot s_i}_{\text{small noise}}
  \quad(\bmod\,q)
  \;\approx\; m.
\]
All operations above are $(\bmod\,q)$.
\item \textbf{Additive Homomorphism:}
Given two ciphertexts:
\[
  (c_0^i,c_1^i)=\text{Enc}_{\text{pk}_i}(m_i),
  \quad
  (c_0^j,c_1^j)=\text{Enc}_{\text{pk}_j}(m_j),
\]
the algorithm forms:
\[
  C_0 = c_0^i + c_0^j \quad(\bmod\,q),
  \quad
  C_1 = c_1^i + c_1^j \quad(\bmod\,q),
\]
which we encrypt $m_i+m_j$, since
\begin{align*}
    C_0 + C_1\cdot s_i + C_1\cdot s_j 
  &= (c_0^i+c_0^j)
    + (c_1^i+c_1^j)(s_i + s_j)\\
  &\approx m_i + m_j
  \quad(\bmod\,q).\\
\end{align*}

\end{itemize}

\subsubsection{Elliptic Curve Diffie-Hellman (ECDH)}
ECDH is a key exchange protocol used to generate pairwise masks. Each client $U_i$ generates a private scalar $sk_i^{ecdh}$ and a public key $pk_i^{ecdh} = sk_i^{ecdh} \cdot G$, where $G$ is a public generator point. Two clients, $U_i$ and $U_j$, can compute a shared secret $K_{ij} = sk_i^{ecdh} \cdot pk_j^{ecdh} = sk_j^{ecdh} \cdot pk_i^{ecdh}$. The security of ECDH relies on the Computational Diffie-Hellman (CDH) Assumption, which states that it is infeasible to compute $K_{ij}$ given only the public keys $pk_i^{ecdh}$ and $pk_j^{ecdh}$.
\section{Hyb-Agg: A New Secure Aggregation Scheme for FL}

We consider an MK\text{-}CKKS–based design in which each client \(U_i\) sends its
encrypted update \(\mathbf{c}^i=(c_0^i,c_1^i)\) and a partial decryption share
\(\mu_i = c_1^i s_i + e_i^*\,(\bmod\,q\)) to the server. The server computes
\(\sum_i c_0^i\) and \(\sum_i \mu_i\) and recovers
\(\sum_i m_i \approx \sum_i c_0^i + \sum_i \mu_i\,(\bmod\,q\).
However, since \(c_0^i + \mu_i \approx m_i\,(\bmod\,q)\), the server could decrypt
\(m_i\). We therefore mask \(\mu_i\) with a random \(r_i\) so that no entity
holds \(c_0^i\) and \(\mu_i\) together.

\subsection{Construction}
\begin{mybox}{\textsf{\textbf{Protocol} GenerateSetup}}
\textbf{Require:} All clients available.

\begin{algorithmic}[1]
  \State Server publishes parameters \((n,q,\chi,\psi,a)\).
  \State Each \(U_i\) generates MK\!-CKKS keys \((s_i,b_i)\) with \(b_i=-s_i a + e_i \bmod q\),
         \(e_i \gets \psi\), and ECDH keys \((sk_i^{\mathrm{ECDH}}, pk_i^{\mathrm{ECDH}})\).
\end{algorithmic}

\textbf{Output:} The server distributes all public keys \(\{(b_i, pk_i^{\text{ECDH}})\}\) among clients.
\end{mybox}

\begin{mybox}{\textsf{\textbf{Protocol} Encryption and Masking}}
\begin{algorithmic}[1]
            \State  Client \(U_i\) encode the local model update $x=(x_0,\dots,x_{d-1})\in\mathbb R^d$ by invoking the function:
                \[
              m(X) 
              = \sum_{j=0}^{d-1}\bigl\lfloor x_j\,\Delta\bigr\rfloor\,X^j
              \;\in \mathcal R_q,
            \]
            to produce to plaintext polynomial \(m_i \in R_q\).
            \State  Client \(U_i\) encrypts \(m_i\) under its public key:
              \[
                c_0^i = v_i\cdot b_i + m_i + e_0^i \pmod q,
              \]
              \[
                \ c_1^i = v_i\cdot a + e_1^i \pmod q,
              \]
      with random noise \(v_i \xleftarrow{\$}\chi\), \(e_0^i, e_1^i \xleftarrow{\$}\psi\).
      \State  Client \(U_i\) continuously computes:
      \[
        \mu_i = c_1^i \cdot s_i + e_i^* \pmod q,\quad e_i^*\xleftarrow{\$}\phi,
      \]
      where \(\phi\) has larger variance for security. To protect \(\mu_i\), each client generates a random mask \(r_i\). This is achieved by computing a pairwise shared secret \(K_{ij}\) with every other client \(U_j\) using ECDH. Each secret seeds a pseudorandom generator (PRG), based on the ChaCha20 stream cipher\cite{Ber08a}, to derive a shared random polynomial \(p_{ij}\). The final mask is computed as \(r_i = \sum_{j > i} p_{ij} - \sum_{j < i} p_{ji} \pmod q\), which ensures that \(\sum_{i} r_i = 0 \pmod q\). The masked share is:
    \[
      \widetilde{\mu}_i = \mu_i + r_i \pmod q.
    \]

        \end{algorithmic}
        \textbf{Output}: Client \(U_i\) send only \((c_0^i,\widetilde{\mu}_i)\) to the server.
    \end{mybox}
\begin{mybox}{\textsf{\textbf{Protocol} Aggregation and Recovery}}
\textbf{Require}:  all masked shares and ciphertexts from clients \(\{U_i\}_{i = 1}^n\)
        \begin{algorithmic}[1]
            \State Server computes:
              \begin{align*}
                  C_0 &= \sum_{i} c_0^i \pmod q \\
                  \widetilde{M} &= \sum_{i}\widetilde{\mu}_i = \sum_{i}(\mu_i + r_i) = \sum_{i}\mu_i \pmod q.
              \end{align*}
              \State   Due to mask cancellation, the server recovers the plaintext sum:
              \[
                \sum_{i} m_i \approx C_0 + \widetilde{M} \pmod q.
              \]
            \State Server server decodes the aggregated plaintext polynomial to obtain the final aggregated update:
              \[
                \sum_{i} x_i = \text{Decode}\left(\sum_{i} m_i\right).
              \]
              by applying the Decode function: \[
              \tilde x_j = \frac{\mathsf{Coef}_j(m)}{\Delta},
              \quad j=0,\dots,d-1,
            \]
        \end{algorithmic}
        \textbf{Output}: Server recovers the sum values of all clients as \(\sum_{i} x_i\)
    \end{mybox}
\section{Security Analysis}

\begin{theorem}
Under the RLWE and CDH assumptions, and modeling the ChaCha20 PRG as a pseudorandom function, there is no probabilistic polynomial-time (PPT) adversary, whether server or client, that can distinguish between real protocol executions and simulated executions with random inputs.
\end{theorem}

We prove Theorem 1 by first establishing Lemma 1 and Lemma 2, from which the theorem follows as a direct consequence. Based on this strategy, our proof demonstrates that the protocol achieves confidentiality, integrity, and collusion resistance under standard cryptographic assumptions.

\begin{lemma}
    Given a secure cryptographic scheme, the Encryption and Masking protocol is secure. Under the RLWE and CDH assumptions, and modeling the ChaCha20 PRG as a pseudorandom function, no probabilistic polynomial-time (PPT) adversary—whether server or client—can distinguish between real protocol executions and simulated executions with random inputs.
\end{lemma}
\begin{proof}
Let $\mathcal{A}$ be any PPT adversary. We consider both server and client adversaries simultaneously, as they share the same protocol view:

\emph{Common Adversarial View:}
\[
\mathcal{V} = \left( \{pk_i^{\text{ECDH}}\}_{i=1}^n, \{pk_i^{\text{MK-CKKS}}\}_{i=1}^n, \{(c_0^{i}, \tilde{\mu}_i)\}_{i=1}^n \right)
\]
where $c_0^{i} = \text{Enc}(pk_i, \mathbf{m}_i)$ and $\tilde{\mu}_i = \mu_i + \mathbf{r}_i$.\\

\begin{enumerate}[leftmargin=*]
    \item \textbf{Ciphertext Indistinguishability}:
    For each client $U_i$, the MK-CKKS ciphertext is computed as:
    \[
    \text{Enc}(pk_i, \mathbf{m}_i) \equiv (v_i \cdot b_i + \mathbf{m}_i + e_0^{i}, v_i \cdot a + e_1^{i}) \pmod q
    \]
    Under the RLWE assumption, for random $v_i, e_0^{i}, e_1^{i} \leftarrow \chi$, the ciphertext is computationally indistinguishable from a pair of uniformly random polynomials $(u_0, u_1) \in \mathcal R_q^2$. Therefore:
    \[
    \left| \Pr[\mathcal{A}(\text{Enc}(pk_i, \mathbf{m}_i)) = 1] - \Pr[\mathcal{A}(u_0, u_1) = 1] \right| \leq \text{negl}(\lambda)
    \]

    \item \textbf{Mask Obfuscation}:
    Each additive mask \(\mathbf{r}_i\) is the sum and difference of multiple polynomials \(p_{ij}\), as described in Section IV-A. Each \(p_{ij}\) is generated by a PRG (ChaCha20) seeded with a shared secret \(K_{ij}\) derived from an ECDH key exchange.
    \[
    \mathbf{p}_{ij} = \text{ChaCha20-PRG}(K_{ij})
    \]
    For any non-colluding pair $(U_i, U_j)$, the ECDH secret $K_{ij}$ is computationally hard to determine from the public keys \(pk_i^{\text{ECDH}}\) and \(pk_j^{\text{ECDH}}\) under the **Computational Diffie-Hellman (CDH) assumption**. Modeling **ChaCha20 as a secure pseudorandom function (PRF)**, its output \(p_{ij}\) is computationally indistinguishable from a truly random polynomial. Since \(\mathbf{r}_i\) is a linear combination of these pseudorandom polynomials, \(\mathbf{r}_i\) itself is computationally indistinguishable from a random polynomial in \(\mathcal R_q\).

    \item \textbf{Hybrid Argument}:
    Construct a series of hybrid games:
    \begin{itemize}
        \item $\mathsf{Hyb}_0$: The real protocol execution.
        \item $\mathsf{Hyb}_1$: Replace all ciphertexts $(c_0^{i}, c_1^{i})$ with encryptions of zero. This is indistinguishable from \(\mathsf{Hyb}_0\) by the semantic security of MK-CKKS, which relies on the RLWE assumption.
        \item $\mathsf{Hyb}_2$: Replace the output of the ChaCha20 PRG for each pairwise secret with truly random polynomials. This is indistinguishable from \(\mathsf{Hyb}_1\) under the assumption that ChaCha20 is a secure PRF.
    \end{itemize}
    In \(\mathsf{Hyb}_2\), the masked share \(\tilde{\mu}_i\) is a sum of a component related to an encryption of zero and a truly random polynomial, which is itself indistinguishable from random. As the adversary's view consists entirely of random-looking values, it can gain no advantage. The indistinguishability of these hybrids implies the security of the protocol.
\end{enumerate}
\end{proof}

\begin{lemma}[Collusion Resistance]\label{thm:collusion}
A coalition of $k \leq n-2$ clients and the server cannot recover any honest client's model update $\mathbf{x}_h$, provided there are at least two honest clients remaining.
\end{lemma}

\begin{proof}
Let $T \subset [n]$ with $|T| = k \leq n-2$ be the set of colluding clients. Let $H$ be the set of honest clients, where $|H| \geq 2$. For any honest client $U_h \in H$:

\begin{enumerate}[leftmargin=*]
    \item \textbf{Residual Mask Analysis}:
    The adversary knows the private keys of all clients in $T$. The mask \(\mathbf{r}_h\) for an honest client \(U_h\) is composed of two parts:
 \begin{align*}
\mathbf{r}_h ={} & \underbrace{\sum_{j \in T} \sigma_{hj} \cdot \text{PRG}(K_{hj})}_{\text{Known to Adversary}} 
& + \underbrace{\sum_{j \in H, j \neq h} \sigma_{hj} \cdot \text{PRG}(K_{hj})}_{\text{Residual Mask}}
\end{align*}
    The adversary can compute the first term because it possesses the secret keys of all clients in $T$ and can therefore compute \(K_{hj}\) for all \(j \in T\). However, the residual mask depends on shared secrets \(K_{hj}\) where both \(h, j \in H\). Since there is at least one other honest client \(h' \in H\), the residual sum includes the term \(\text{ChaCha20-PRG}(K_{hh'})\). Under the CDH assumption, the adversary cannot compute the secret \(K_{hh'}\).

    \item \textbf{Indistinguishability Preservation}:
    Because the residual mask term is a sum that includes at least one output from the PRG seeded with an unknown secret, the entire residual mask remains pseudorandom from the adversary's perspective.
    \[
    \mathbf{r}_h^{\text{res}} = \sum_{j \in H, j \neq h} \sigma_{hj} \cdot \text{ChaCha20-PRG}(K_{hj}) \approx_c \mathbf{r}'_h \stackrel{\$}{\leftarrow} \mathcal R_q
    \]
    Thus, the complete mask \(\mathbf{r}_h\) remains computationally indistinguishable from a random polynomial, even with the knowledge of the colluding parties.
\end{enumerate}    
\end{proof}
   
\textbf{Aggregation Security}:
    While colluders know their own updates \(\sum_{i \in T} \mathbf{x}_i\), the server only ever sees the final aggregate sum:
    \[
    \sum_{i=1}^n \mathbf{m}_i = \sum_{i \in T} \mathbf{m}_i + \sum_{i \in H} \mathbf{m}_i
    \]
    Because there are at least two honest parties (\(|H| \geq 2\)), the term \(\sum_{i \in H} \mathbf{m}_i\) is a sum of at least two unknown plaintexts. The individual ciphertexts and masked shares of honest clients are protected by RLWE-based encryption and pseudorandom masks, respectively. Therefore, the coalition cannot isolate or decrypt any individual honest client's update.

This secure aggregation mechanism ensures that the server or colluding clients never observe individual plaintext updates, thus providing strong privacy guarantees without compromising the accuracy of the federated learning model aggregation. The proofs for Lemma 1 and Lemma 2 confirm Theorem 1, establishing that Hyb-Agg is secure under the RLWE and CDH assumptions, and by modeling ChaCha20 as a secure PRF. $\Box$
\section{Performance Evaluation}
\label{sec:evaluation}

This section presents a comprehensive empirical evaluation of the proposed Hyb-Agg protocol. The primary objective is to validate the theoretical performance claims and assess its practical feasibility, with a specific focus on IoT environments. To this end, our evaluation includes benchmarks on two distinct hardware platforms: a high-performance machine to establish a baseline, and a representative resource-constrained device (Raspberry Pi 4) to demonstrate viability at the edge.

The analysis is structured as follows. We begin by detailing the experimental methodology, including the simulation environments and cryptographic parameterization. We then provide a detailed theoretical analysis of the protocol's computational, communication, and storage complexity. Subsequently, we analyze the protocol's numerical precision, a critical factor for any scheme based on approximate homomorphic encryption like CKKS~\cite{cheon2017ckks}. We then present and analyze the empirical results, analyzing scalability with respect to both client population size ($N$) and data vector dimensionality ($d$). We conclude with a discussion that contextualizes these findings and outlines the inherent design trade-offs of the Hyb-Agg approach.

\subsection{Experimental Methodology}
\label{sec:setup}

\subsubsection{Environments and Implementation}
To provide a comprehensive view, two distinct sets of experiments were conducted. The first set was run on a high-performance laptop to serve as a baseline. The second, more critical set for our use case, involved running the \textit{entire} simulation including both the server and all client processes on a Raspberry Pi 4 to directly assess end-to-end performance in a resource-constrained environment.

\paragraph{High-Performance Environment}
The baseline experiments were executed on a laptop equipped with an AMD Ryzen 9 6900HS processor (3.30 GHz) and 16.0 GB of RAM. The software environment consisted of Ubuntu 24.04 running under the Windows Subsystem for Linux 2 (WSL2).

\paragraph{Resource-Constrained Environment}
To validate performance for IoT scenarios, the complete simulation was also deployed and executed on a Raspberry Pi 4 Model B. The device specifications are as follows:
\begin{itemize}
    \item \textbf{Processor:} Broadcom BCM2711, Quad-core Cortex-A72 (ARM v8) 64-bit SoC @ 1.5 GHz.
    \item \textbf{RAM:} 8GB LPDDR4-3200 SDRAM.
    \item \textbf{Storage:} Micro SD card for the operating system and data.
\end{itemize}

\paragraph{Implementation Details}
The cryptographic primitives that form the foundation of Hyb-Agg were implemented in C++. Specifically, the Multi-Key CKKS (MK-CKKS) scheme~\cite{Chen2019MKCKKS} was implemented using the open-source \textbf{OpenFHE library}~\cite{openfhe_lib}. All experiments were conducted in a simulated environment on a single machine (either the laptop or the Raspberry Pi) to ensure consistent and reproducible measurements of cryptographic and computational overhead, thereby isolating these factors from the variability of network latency.

\subsubsection{Cryptographic Parameterization}
The performance and security of the CKKS scheme are governed by a set of carefully chosen parameters. For all experiments, the parameters were selected to provide a standard \textbf{128-bit security level} against both classical and quantum adversaries, as recommended by established FHE security standards~\cite{cryptoeprint:2019/939}.

\begin{itemize}
    \item \textbf{Ring Dimension ($n$):} The ring dimension $n$ is a power-of-two integer that defines the degree of the polynomial ring $\mathbb{R}_q = \mathbb{Z}_q[X]/(X^n+1)$ used for all cryptographic operations~\cite{openfhe_lib}. It is a primary driver of both performance and ciphertext capacity. In our experiments, $n$ was chosen dynamically by the OpenFHE library. For a given data vector of dimension $d$, the library selects the smallest power of two $n$ such that the number of available plaintext slots, $n/2$, is sufficient to encode the vector (i.e., $n/2 \ge d$), while simultaneously satisfying the 128-bit security requirement. This discrete selection process is the cause of the "step-function" behavior observed in our performance results.
    \item \textbf{Scaling Factor and Moduli:} The CKKS scheme performs approximate arithmetic on real numbers by encoding them as scaled integers~\cite{cheon2017ckks}. We configured OpenFHE with a \texttt{multiplicativeDepth} of 1, as our aggregation protocol only requires homomorphic additions; additions contribute far less noise than multiplications and OpenFHE sizes the modulus based on depth and the scaling factor~\cite{openfhe_aaai_tutorial}. The \texttt{scaleModSize} was set to 50 bits. This parameter sets the bit-length used for the scaling/rescaling levels in the CKKS modulus chain and thereby supports a large scaling factor $\Delta$ (typically $\Delta=2^{\text{scaleModSize}}$), yielding high-precision encodings~\cite{openfhe_docs_params,gharibyar2025ckksparams}.

\end{itemize}

\subsubsection{Data Generation}
The experiments were conducted using \textbf{synthetic data}. Specifically, client inputs were modeled as vectors of randomly generated 64-bit floating-point numbers (double), with the dimensionality $d$ varied as an experimental parameter. The choice of synthetic data is intentional and justified by the scope of our evaluation. The primary goal is to measure the performance overhead of the cryptographic protocol itself, not the efficacy of a specific machine learning model. The computational and communication costs of our protocol are dependent on the \textit{dimensionality} and \textit{bit-length} of the client data vectors, not on their statistical properties. Therefore, using synthetic data allows for a clean and reproducible measurement of the protocol's scalability, a standard practice when benchmarking the fundamental costs of privacy-enhancing technologies~\cite{royalsoc2023synthetic, cormode2025synthetic}.

\subsection{Theoretical Performance Analysis}
We first analyze the asymptotic complexity of Hyb-Agg in terms of computation, communication, and storage for both clients and the server, where $N$ is the number of clients and $d$ is the dimension of the model update vector.

\subsubsection{Client-Side Complexity}
\begin{itemize}
    \item \textbf{Computation ($O(N \cdot d)$):} The total computational work for each client is the sum of the costs of encoding ($O(d)$), encryption ($O(d \log d)$), partial decryption share generation ($O(d \log d)$), and pairwise mask generation. The dominant step with respect to $N$ is the mask generation, which involves an ECDH key exchange with each of the other $N-1$ clients and generating a mask of dimension $d$ for each. This results in a complexity of $O(N \cdot d)$, making it the overall asymptotic complexity for the client.
    \item \textbf{Communication ($O(d)$ per round):} During the one-time setup phase, each client receives the public keys of all other $N-1$ clients, resulting in an initial communication cost of $O(N)$. However, in each subsequent learning round, a client sends only a single payload to the server, consisting of the components ($c_{0}^{i}, \tilde{\mu}_{i}$). The size of this payload is determined by the cryptographic parameters and the vector dimension $d$, but is independent of the number of clients $N$. Therefore, the crucial per-round communication cost is $O(d)$.
    \item \textbf{Storage ($O(N)$):} Each client must store the public keys (both MK-CKKS and ECDH) of all other $N-1$ clients to generate the pairwise masks. This leads to a storage requirement that is linear in the number of clients.
\end{itemize}

\subsubsection{Server-Side Complexity}
\begin{itemize}
    \item \textbf{Computation ($O(N \cdot d)$):} The server's workload consists of aggregating ciphertexts and masked shares, followed by recovery and decoding. The dominant task is summing the vector components received from all $N$ clients. This involves $N-1$ additions for each of the two components, resulting in a total complexity of $O(N \cdot d)$.
    \item \textbf{Communication ($O(N \cdot d)$ per round):} In each round, the server receives one payload of size $O(d)$ from each of the $N$ clients. The total uplink communication cost for the server is therefore linear in both the number of clients and the data dimension, with a complexity of $O(N \cdot d)$. The downlink communication is minimal, consisting only of the final aggregated result of size $O(d)$.
    \item \textbf{Storage ($O(d)$):} The server does not need to store individual client submissions. It only needs to maintain the running sums of the received vector components, requiring storage proportional to the vector dimension $d$.
\end{itemize}

\subsection{Analysis of Protocol Accuracy and Precision}
A crucial consideration for any system employing approximate homomorphic encryption is the impact of inherent computational noise on the final result~\cite{cheon2017ckks,lee2023elasm,cryptoeprint:2019/939}. In our experiments, we observed that the decrypted aggregate was numerically identical to the true sum of the plaintext vectors, resulting in an effective approximation error of zero.

This outcome is not a contradiction of the approximate nature of CKKS but is rather a direct consequence of the specific application (secure summation) combined with a robust parameterization. The noise introduced by homomorphic \emph{addition} grows much more slowly than that from multiplication~\cite{cheon2017ckks,lee2023elasm}. With a \texttt{scaleModSize} of 50 bits, OpenFHE’s CKKS configuration uses a large scaling factor $\Delta$ (on the order of $2^{50}$), which, together with depth-1 circuits, keeps the accumulated addition noise well below the decoding threshold~\cite{openfhe_docs_params,openfhe_aaai_tutorial,gharibyar2025ckksparams}. The final decoding step divides by $\Delta$ and rounds, effectively quantizing this negligible error to zero~\cite{cheon2017ckks}.

\subsection{Empirical Performance Evaluation}
\label{sec:empirical_eval}

\subsubsection{Communication Overhead}
Communication efficiency is a paramount concern in federated learning. Our empirical results validate the theoretical analysis. As shown in \textbf{Fig. \ref{fig:client_comm_vs_clients}}, a key architectural strength of Hyb-Agg is that the per-round communication cost per client is constant with respect to the number of participating clients, $N$. For a fixed data vector size, the client uplink remains flat, confirming the $O(d)$ per-round complexity.

Conversely, \textbf{Fig. \ref{fig:client_expansion}} shows that the communication cost scales with the dimension of the data vector, $d$. The characteristic "step" pattern reflects the underlying cryptographic parameter changes, where the ring dimension $n$ must increase to accommodate a larger $d$.

To quantify the overhead, we analyze the communication expansion factor. For a data vector of dimension $d=65,536$ composed of 64-bit (8-byte) floating-point numbers, the plaintext size is $65,536 \times 8 = 524,288$ bytes. As shown in \textbf{Fig. \ref{fig:client_expansion}}, the corresponding `ClientUplinkBytes` is approximately 6.3MB. This yields a communication expansion factor of approximately $6,300,000 / 524,288 \approx \mathbf{12\times}$. This overhead is the explicit trade-off for achieving strong privacy guarantees with a non-interactive protocol.

\textbf{Note on expansion vs.\ data size.}
The expansion factor primarily reflects CKKS packing efficiency (slots $= n/2$).
When vectors use most slots, overhead is amortized and we observe $\approx 12\times$ (e.g., $d=65{,}536$ with $n=131{,}072$);
with smaller $d$, unused slots let fixed polynomial/modulus bytes dominate, so the factor can be higher (about $24\times$ near $d\approx 4{,}095$), consistent with the “step” behavior in Fig.~2.
In practice, simple packing—grouping multiple smaller subvectors (layers/blocks) into one plaintext—or choosing parameters that better match $d$ keeps results close to the $\approx 12\times$ regime while preserving the target security level.

\begin{figure}[htbp]
\centering
\includegraphics[width=\columnwidth]{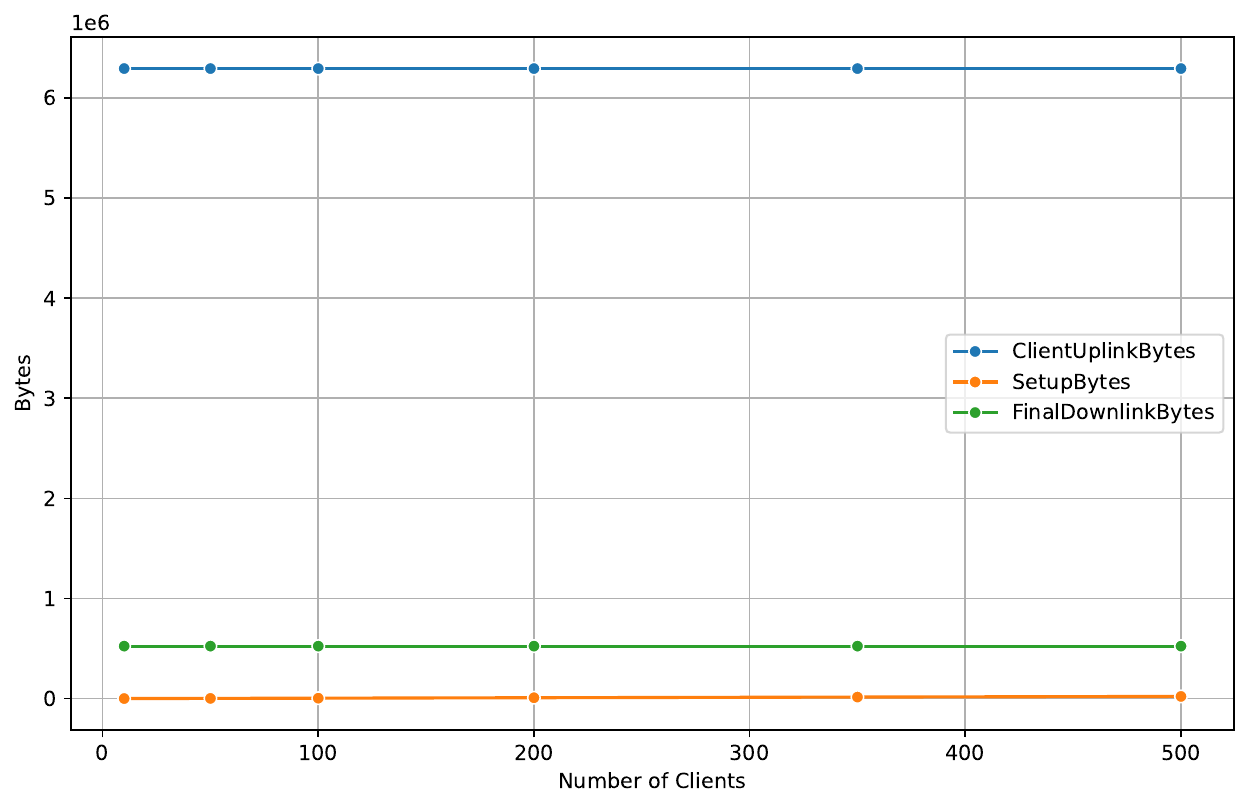}
\caption{Byte Sizes vs. Number of Clients (DataSize = 65536). This figure illustrates how various communication byte sizes per client change as the number of clients increases, while the data size is kept constant at 65536.}
\label{fig:client_comm_vs_clients}
\end{figure}

\begin{figure}[htbp]
\centering
\includegraphics[width=\columnwidth]{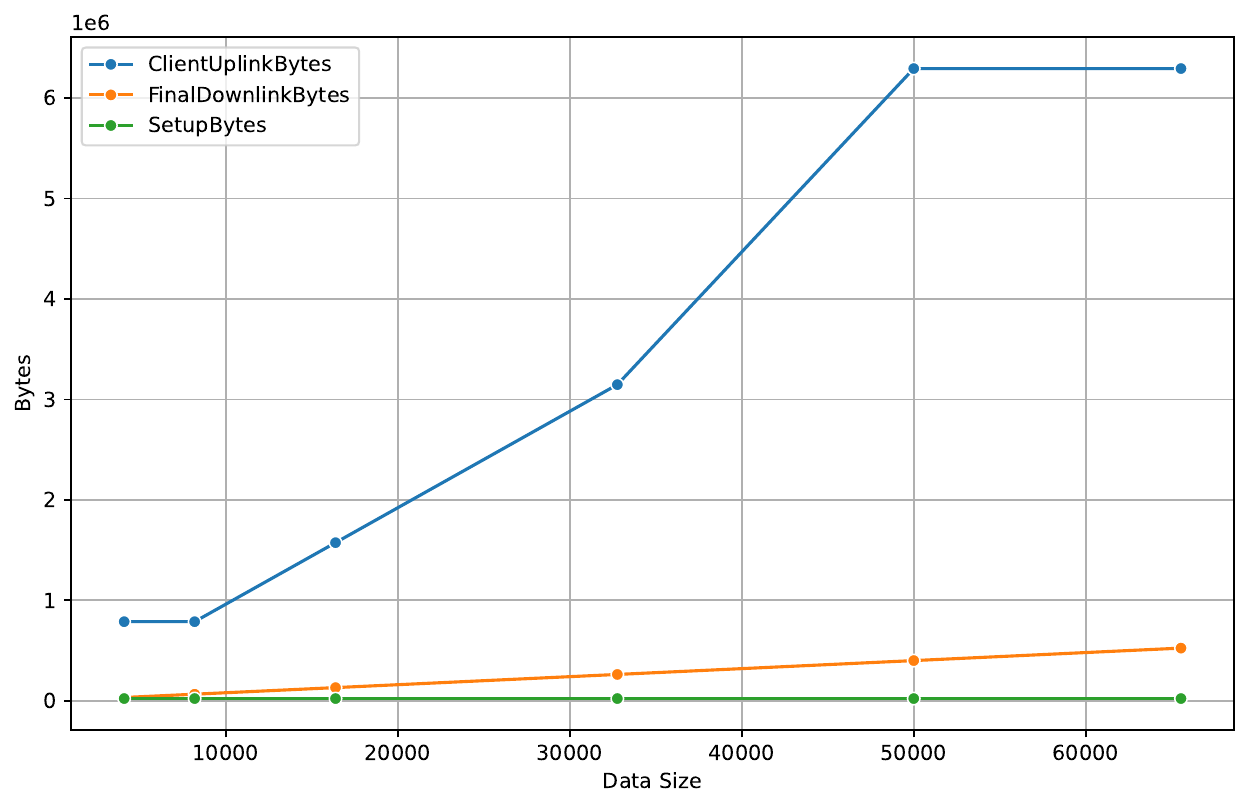}
\caption{Byte Sizes vs. Data Size (NumClients = 500). This figure shows how the client uplink, final downlink, and setup byte sizes vary with different data vector sizes, for a fixed number of clients (500)}
\label{fig:client_expansion}
\end{figure}

\subsubsection{Computational Performance on High-Performance Hardware}
\textbf{Client-Side Computation:} As predicted by our theoretical analysis and shown in \textbf{Fig. \ref{fig:client_time_vs_clients}}, the total client-side time on our test laptop scales linearly with the number of clients $N$. The detailed breakdown reveals this is driven entirely by the pairwise mask generation. \textbf{Fig. \ref{fig:client_time_vs_size}} illustrates the impact of data dimensionality, where the distinct step-function behavior is caused by the FHE library selecting a larger ring dimension $n$ to accommodate the increased data vector size.

\begin{figure}[htbp]
\centering
\includegraphics[width=\columnwidth]{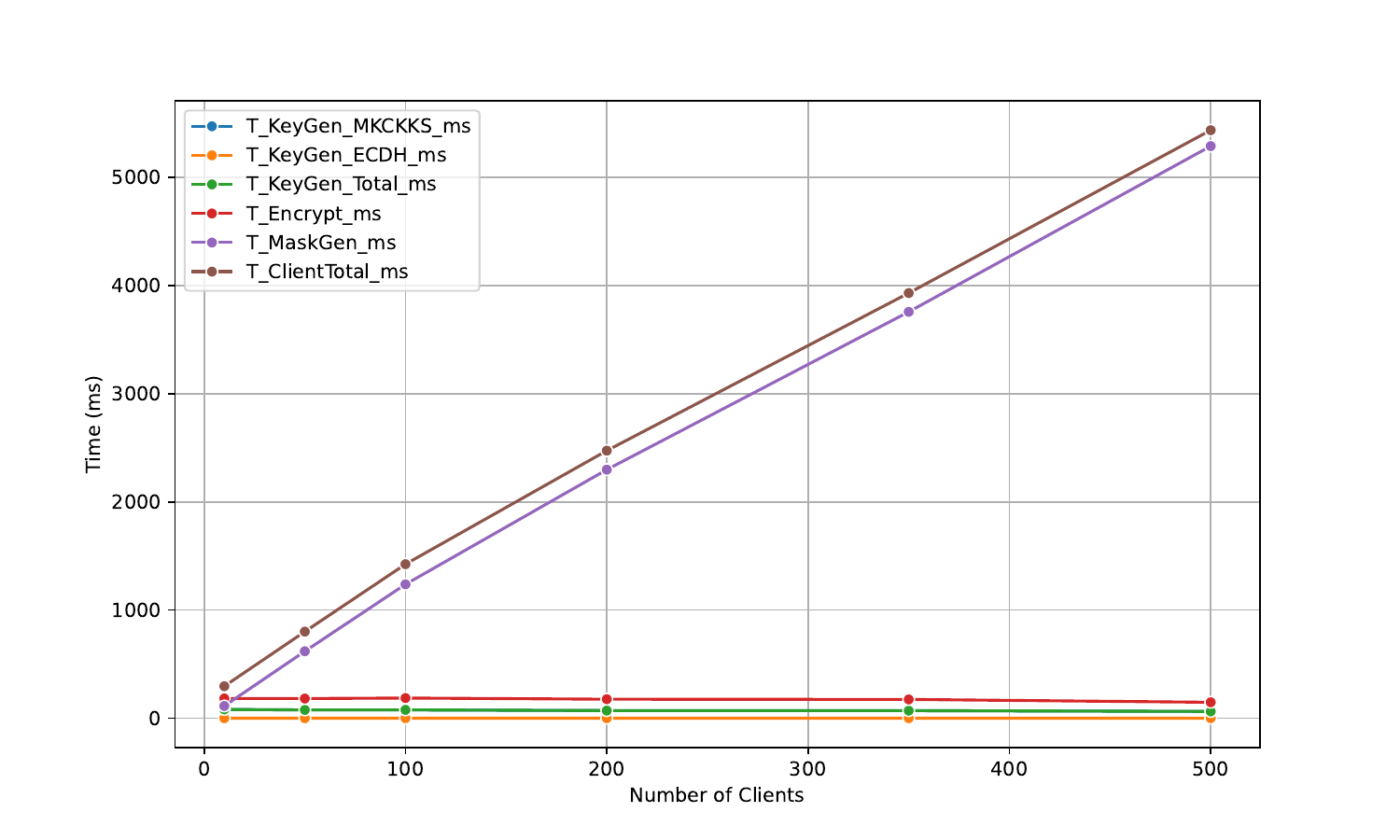}
\caption{Wall-clock running time per client vs. number of clients for Hyb-Agg. The data vector size is fixed to 65,536.}
\label{fig:client_time_vs_clients}
\end{figure}

\begin{figure}[htbp]
\centering
\includegraphics[width=\columnwidth]{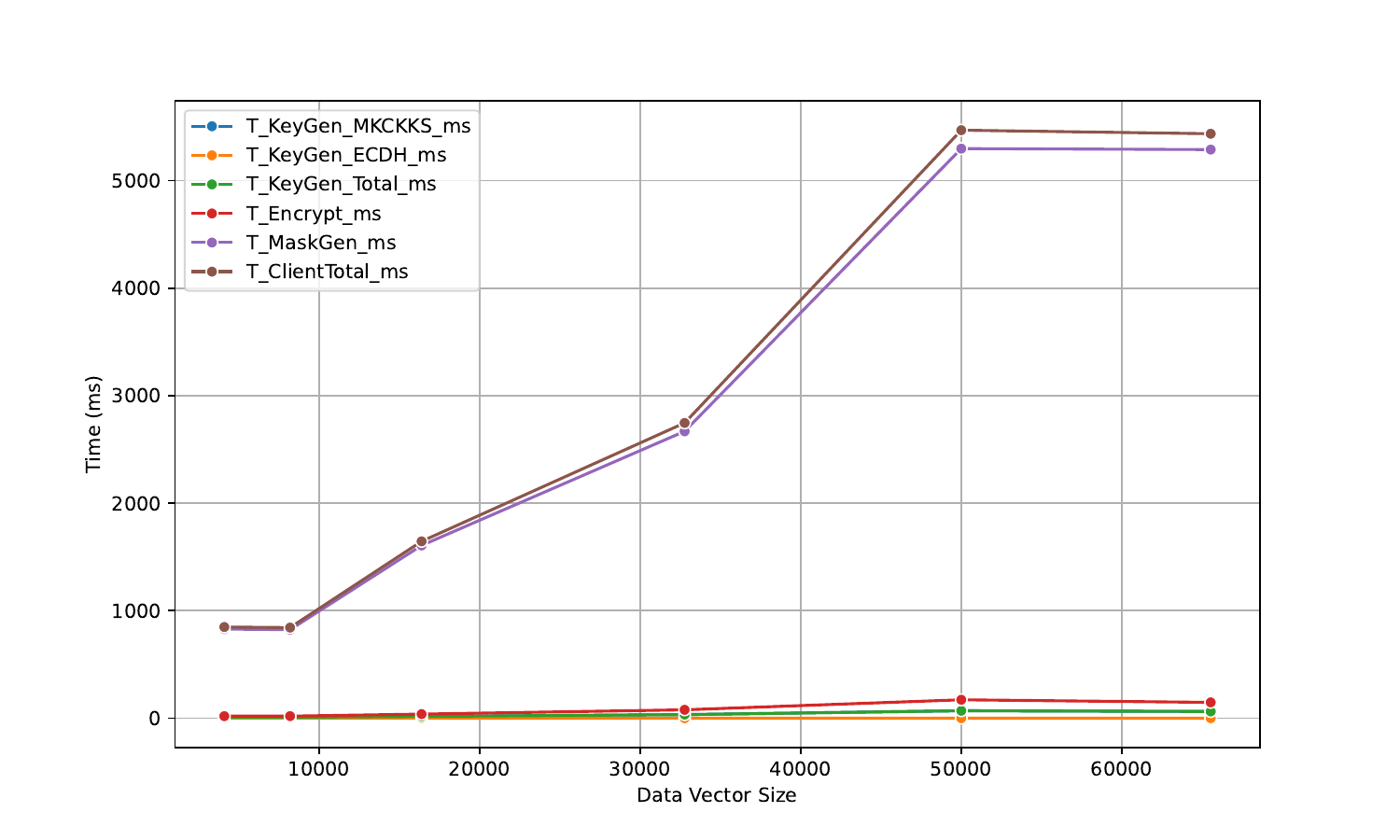}
\caption{Wall-clock running time per client vs. data vector size for all protocols. The number of clients is fixed at 500.}
\label{fig:client_time_vs_size}
\end{figure}

\textbf{Server-Side Computation:} The server's computational load, depicted in \textbf{Fig. \ref{fig:server_time_vs_clients}}, scales linearly with the number of clients $N$, as theoretically predicted. \textbf{Fig. \ref{fig:server_time_vs_size}} reveals a more complex, non-monotonic relationship between server time and data size. This behavior is also an artifact of the FHE parameterization.

\begin{figure}[htbp]
\centering
\includegraphics[width=\columnwidth]{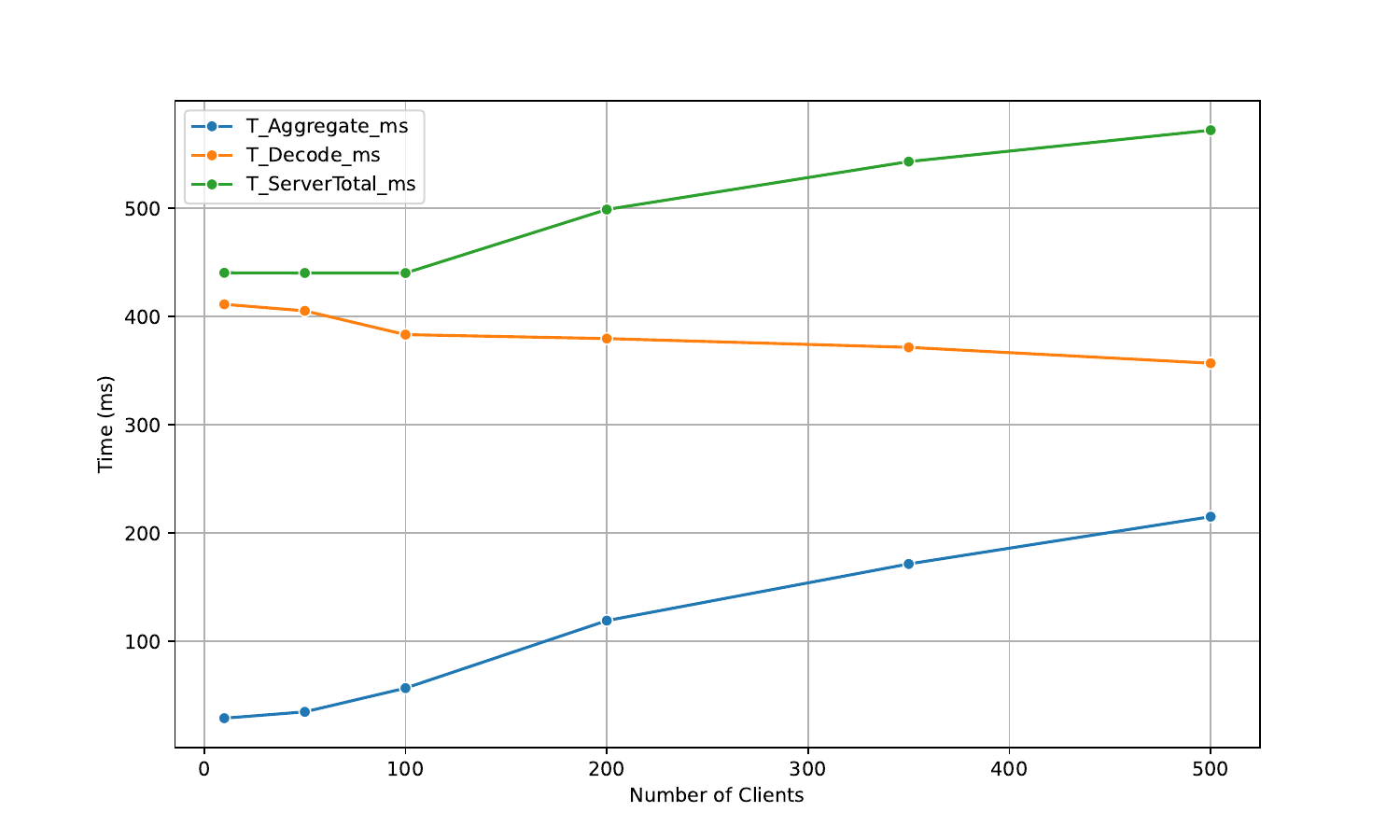}
\caption{Wall-clock running time for the server vs. number of clients for all protocols. The data vector size is fixed to 65,536.}
\label{fig:server_time_vs_clients}
\end{figure}

\begin{figure}[htbp]
\centering
\includegraphics[width=\columnwidth]{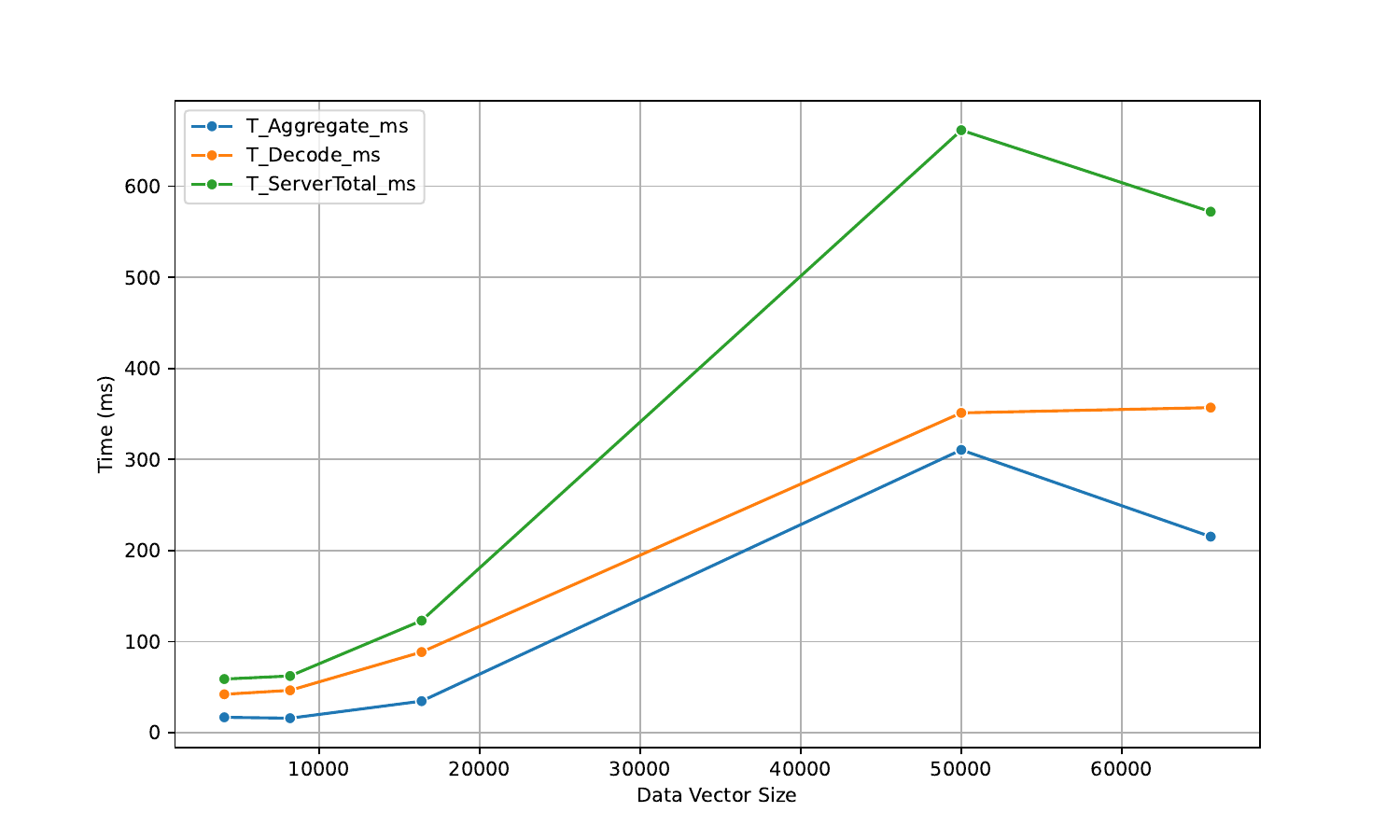}
\caption{Wall-clock running time for the server vs. data vector size for all protocols. The number of clients is fixed at 500.}
\label{fig:server_time_vs_size}
\end{figure}

\subsubsection{End-to-End Performance on a Resource-Constrained Device}
To directly validate the feasibility of Hyb-Agg in IoT environments, we executed the entire simulation on a Raspberry Pi 4. Our results for a representative workload are summarized in Table \ref{tab:rpi_results}.

\begin{table}[h]
\centering
\caption{Hyb-Agg Performance on Raspberry Pi 4 ($N$=50, $d$=8192)}
\label{tab:rpi_results}
\begin{tabular}{lc}
\toprule
\textbf{Metric} & \textbf{Measured Performance} \\
\midrule
Avg. Client Total Time & 431.8 ms \\
Server Total Time & 191.1 ms \\
Client Uplink per Round & 787 KB \\
Communication Expansion & $\approx$12x \\
\bottomrule
\end{tabular}
\end{table}

The data demonstrates the protocol's practicality on edge hardware. For a simulation of 50 clients with a data vector size of 8192, the average total time for a client to perform its key generation, encryption, and masking operations was approximately 431.8 ms. The server, also running on the Pi, completed its aggregation and decoding tasks in just 191.1 ms.

These results are highly encouraging. The sub-second execution times for both client and server workloads are well within practical limits for many real-world FL applications, particularly those where training rounds occur infrequently (e.g., once every few hours or daily). This experiment provides strong empirical evidence that Hyb-Agg is not merely a theoretical construct but a viable protocol for securing federated learning end-to-end on the types of resource-constrained devices that characterize the Internet of Things.

\subsection{Discussion and Comparative Analysis}
\subsubsection{Summary of Findings}
The empirical evaluation demonstrates that Hyb-Agg offers a compelling performance profile. Its primary advantage is the per-client communication cost, which is constant with respect to the number of clients, making it highly scalable. This is achieved at the cost of client-side computation that is linear in the number of clients ($O(N \cdot d)$) and a constant data expansion factor of approximately 12x. The protocol achieves practical exactness for aggregation, confirming its numerical integrity.

\subsubsection{Robustness to Client Dropouts}
The experiments presented assume a scenario with no client dropouts to establish a clear performance baseline. It is acknowledged that handling dropouts is critical for practical FL systems. The Hyb-Agg protocol can be readily extended to support this. The additive masking scheme is conceptually similar to the foundation of the protocol by Bonawitz et al.~\cite{bonawitz2017ccs}. Therefore, a robust dropout handling mechanism could be integrated by having clients share their pairwise mask seeds using Shamir's Secret Sharing (SSS)\cite{Shamir1979CACM}. This would introduce additional communication and computation for reconstruction, representing a well-understood engineering trade-off and a clear direction for future work.

\subsubsection{Conclusion on Performance Trade-offs}
In summary, Hyb-Agg carves out a specific and valuable niche in the design space of secure aggregation protocols. It makes a deliberate trade-off, prioritizing \textbf{communication scalability with client population} and \textbf{protocol simplicity} (a single, non-interactive round) by leveraging multi-key homomorphic encryption. The explicit costs are an increased \textbf{client-side computational complexity} and a \textbf{constant data expansion factor}. This profile makes Hyb-Agg an effective solution for secure aggregation in IoT environments. Our end-to-end simulation on a Raspberry Pi 4 substantiates this claim, demonstrating that while the computational load is significant, the absolute performance remains practical for many real-world edge deployment scenarios.

\section{Conclusion}
This work introduced Hyb-Agg, a lightweight and communication-efficient secure aggregation protocol for federated learning in IoT environments. By integrating Multi-Key CKKS homomorphic encryption with ECDH-based additive masking, Hyb-Agg enables one-shot, non-interactive aggregation of model updates while ensuring confidentiality, integrity, and resistance to collusion under standard cryptographic assumptions.

Our theoretical analysis and empirical evaluation demonstrate that Hyb-Agg achieves constant per-client communication cost with respect to the number of participants, sub-second execution times even on resource-constrained devices such as the Raspberry Pi 4, and exact aggregate recovery despite the use of approximate homomorphic encryption. These properties make the protocol particularly well-suited for bandwidth-limited and heterogeneous IoT networks, where scalability and privacy must be balanced against strict performance constraints.

While the current design assumes full client participation and a trusted key setup phase, the protocol can be extended to support dropout resilience and verifiable setup without altering its fundamental one-shot nature. Future work will focus on optimizing mask generation, incorporating dynamic participation models, and evaluating the protocol in real-world wireless IoT deployments.

\bibliographystyle{IEEEtran}
\bibliography{lib}

\end{document}